%% file: ccc-paper.tex
\pdfoutput=1
\documentclass[11pt, fullpage]{article}
\usepackage[utf8]{inputenc}
\usepackage{fullpage}
\makeatletter
\def\input@path{{/home/vlecomte/Dropbox/misc/latex/}}
\makeatother
\usepackage{vlecomte}

\allowdisplaybreaks %

\input{defs}

\begin{document}

\title{The composition complexity of majority
 \vspace{15pt}}

\author{\hspace{-15pt}Victor Lecomte  \and \hspace{-12pt} Prasanna Ramakrishnan\vspace{18pt} \\
 {\sl \small{Stanford University}}
\and \hspace{-10pt} Li-Yang Tan \vspace{8pt} \\ \hspace{-15pt} }

\date{\vspace{15pt}\small{\today}}

\maketitle

\input{abstract}

\newpage

\input{body}

\section*{Acknowledgements}

Li-Yang thanks Xi Chen, Rocco Servedio, and Erik Waingarten for numerous discussions about this problem.  

Victor, Pras, and Li-Yang are supported by NSF CAREER Award CCF-1942123.  Pras is also supported by Moses Charikar's Simons Investigator Award.

\bibliography{pras}
\bibliographystyle{alpha}

\end{document}

%% file: defs.tex
\newcommand{\Jc}{{\overline{J_i}}}
\newcommand{\myw}{{w^*}}
\newcommand{\zo}{\set{0,1}}
\newcommand{\CC}{\mathsf{CC}}
\newcommand{\Ifree}{I_{\mathrm{free}}}
\newcommand{\Icontrol}{I_{\mathrm{control}}}
\newcommand{\Ibuffer}{I_{\mathrm{buffer}}}
\newcommand{\xfree}{x_{\mathrm{free}}}
\newcommand{\xcontrol}{x_{\mathrm{control}}}
\newcommand{\xbuffer}{x_{\mathrm{buffer}}}
\newcommand{\nfree}{n_{\mathrm{free}}}

%% file: abstract.tex
\begin{abstract} 
We study the complexity of computing
majority
as a composition of \emph{local} functions: 
\[ \Maj_n = h(g_1,\ldots,g_m), \] 
where each $g_j : \zo^{n} \to \zo$ is an arbitrary function that queries only $k \ll n$ variables and $h : \zo^m \to\zo$ is an arbitrary
combining function.
We prove an optimal lower bound of 
\[ m \ge \Omega\left( \frac{n}{k} \log k \right) \]
on the number of functions needed, which is a factor $\Omega(\log k)$ larger than the ideal $m = n/k$. We call this factor the \emph{composition overhead}; previously, no superconstant lower bounds on it were known for majority.

Our lower bound recovers, as a corollary and via an entirely different proof, the best known lower bound for bounded-width branching programs for majority~\small{(Alon and Maass '86, Babai et al.~'90)}. It is also the first step in a plan that we propose for breaking a longstanding barrier in lower bounds for small-depth boolean circuits.

Novel aspects of our proof include sharp bounds on the information lost as computation flows through the inner functions $g_j$, and the bootstrapping of lower bounds for a multi-output function (Hamming weight) into lower bounds for a single-output one (majority).

\end{abstract}

%% file: body.tex
\section{Introduction}

\input{intro}

\section{Preliminaries}

\subsection{Locality}
In this subsection, we define some notation that we will use repeatedly, and we give a formal definition of $k$-locality.

\begin{definition}[$x^{(i \mapsto b)}$, $x^{\oplus i}$]
For any $x \in \zo^n$, $i \in [n]$ and $b \in \zo$, let
\begin{itemize}
    \item $x^{(i \mapsto b)} \coloneqq (x_1, \ldots, x_{i-1}, b, x_{i+1}, \ldots, x_n)$;
    \item $x^{\oplus i} \coloneqq (x_1, \ldots, x_{i-1}, 1-x_i, x_{i+1}, \ldots, x_n)$.
\end{itemize}
\end{definition}

\begin{definition}[$k$-local]
A function $g: \zo^n \to \zo$ is $k$-local if there is a set of variables $I \subseteq [n]$ with $|I|=k$ such that $g$ depends only on the variables in $I$, i.e. for each $x\in \zo^n$ and each $i \in [n] \setminus I$, we have $g(x) = g(x^{\oplus i})$.
\end{definition}

\begin{definition}[composition complexity, formal version of \Cref{def:complexity}]
The \emph{$k$-composition complexity} of a function $f: \zo^n \to D$, denoted $\CC_k(f)$, is the minimum integer $m$ such that there exist functions $g_1, \ldots, g_m : \zo^n \to \zo$ and $h:\zo^m \to D$
with the following properties:
\begin{enumerate}[(i)]
    \item for all $x \in \zo^n$, $f(x) = h(g_1(x), \ldots, g_m(x))$;
    \item for each $j \in [m]$, $g_j$ is $k$-local.
\end{enumerate}
\end{definition}

\begin{remark}
Note that the inner functions $g_j$ are restricted to having binary outputs.
This is necessary for making the definition nontrivial: if their output domains were arbitrary, then the composition complexity would always be $O(n/k)$,
since we could simply let each inner function output the values of all of the variables they query.
\end{remark}

\subsection{Number of queries per variable}
In this subsection, we show that without loss of generality, we can assume that all variables are queried roughly the same number of times.
More precisely, say that $f = h(g_1, \ldots, g_m)$ where each $g_j$ is $k$-local.
Then the total number of queries is at most $mk$, so the average variable is queried at most $\frac{mk}{n}$ times.
We will show that for the purpose of proving lower bounds on composition complexity, we can assume that \emph{every} variable is queried at most $\frac{mk}{n}$ times.

\begin{definition}[self-containing]
A family of functions $\set{f_n}_{n \in \N}$ is \emph{self-containing} if for any $n$ and any $I \subseteq [n]$, there is a subfunction of $f_n$ on $I$ that computes $f_{|I|}$.
\end{definition}

\begin{fact}
Both majority $\set{\Maj_n}_{n \in \N}$ and Hamming weight $\set{\HW_n}_{n \in \N}$ are self-containing.%
\end{fact}

\begin{lemma}
\label{lemma:wlog}
Let $\set{f_n}_{n \in \N}$ be a self-containing family of functions.
Suppose that $\CC_k(f_{2n}) \leq m$.
Then we can write $f_n = h(g_1, \ldots, g_m)$ where each $g_j$ is $k$-local and each variable is queried at most $\frac{mk}{n}$ times.
\end{lemma}

\begin{proof}
Suppose $f_{2n} = h(g_1, \ldots, g_m)$, where each $g_j$ is $k$-local.
Let $q \coloneqq \frac{mk}{2n}$.
Then the average variable is queried $\leq q$ times, so by Markov's inequality at most half of the variables are queried more than $2q$ times.
Let $I$ be any set of $n$ variables, each of which is queried at most $2q = \frac{mk}{n}$ times.
Since $\set{f_n}_{n \in \N}$ is self-containing, there is a subfunction of $f_{2n}$ on $I$ that computes $f_n$.
The lemma follows by restricting each $g_j$ to $I$.
\end{proof}

\begin{corollary}
\label{cor:wlog}
In order to prove that $\CC_k(f_n) \geq \Omega(\frac{n}{k} \log k)$, it is enough to prove that if $f_n = h(g_1, \ldots, g_m)$ and each variable is queried at most $\frac{mk}{n}$ times by the inner functions $g_j$, then $m \geq \Omega(\frac{n}{k} \log k)$.
\end{corollary}

\begin{remark}
The quantity $\frac{mk}{n}$ corresponds exactly to our definition of composition overhead (recall the discussion following \Cref{def:complexity}).
This makes sense, since the composition overhead is the ratio of how many inner functions we need compared to the ideal situation where each variable is queried exactly once.
The more inner functions, the more queries per variable (assuming they all query roughly $k$ variables).
\end{remark}

\subsection{Information theory}
In this subsection, we introduce some information-theoretic notions and properties that are used in our proofs.
To learn more about information theory from a theoretical computer science perspective, we recommend checking out the Simons Institute workshop titled ``Information Theory Boot Camp''.\footnote{\url{https://simons.berkeley.edu/workshops/inftheory2015-boot-camp}}

The most important notion in information theory is the \emph{entropy} of a random variable $\bX$, denoted $\H[\bX]$, which represents ``how much randomness'' the variable contains, or how many bits I need to communicate to you on average for you to learn $\bX$.
\begin{definition}[entropy]
Given a random variable $\bX$ with support $D$, the \emph{entropy of $\bX$} is the quantity
\[
\H[\bX] \coloneqq \sum_{x \in D} \Pr[\bX = x] \log \p*{\frac{1}{\Pr[\bX = x]}} = \E_{\bX' \sim \bX}\sqb*{\log \p*{\frac{1}{\Pr[\bX = \bX']}}}
\]
where $\bX'$ is an independent copy of $\bX$.
\end{definition}

A related notion is the \emph{conditional entropy} $\H\sqbcond{\bX}{\bY}$ of two random variables $\bX$ and $\bY$, which represents the ``how much randomness remains'' in $\bX$ once you know $\bY$, or how many bits I need to communicate to you on average for you to learn $\bX$, assuming that you already know $\bY$.

\begin{definition}[conditional entropy]
Given two random variables $(\bX, \bY)$ over domain $D$, the \emph{entropy of $\bX$ conditioned on $\bY$}
is the quantity
\begin{align*}
\H\sqbcond{\bX}{\bY}
&\coloneqq \sum_{(x,y) \in D} \Pr[\bX = x \wedge \bY = y] \log \p*{\frac{1}{\Pr\sqbcond{\bX = x}{\bY = y}}}\\
&= \E_{(\bX',\bY') \sim (\bX,\bY)} \sqb*{\log\p*{\frac{1}{\Pr\sqbcond{\bX = \bX'}{\bY = \bY'}}}}
\end{align*}
where $(\bX',\bY')$ is an independent copy of $(\bX, \bY)$.
\end{definition}

The entropy and conditional entropy have the following properties, which we use in our proofs.

\begin{fact}[bounds]
If $\bX$ is a random variable over a finite domain $D$ and $\bY$ is a random variable, then
\[
0 \leq \H\sqbcond{\bX}{\bY} \leq \H[\bX] \leq \log |D|,
\]
where the last inequality is tight iff $\bX$ is uniform on $D$.
\end{fact}

\begin{fact}[subadditivity]
Let $\bX_1$ and $\bX_2$ be two random variables. Then $\H[\bX_1, \bX_2] \leq \H[\bX_1] + \H[\bX_2]$, with equality iff $\bX_1$ and $\bX_2$ are independent.
Similarly, $\H\sqbcond{\bX_1, \bX_2}{\bY} \leq \H\sqbcond{\bX_1}{\bY} + \H\sqbcond{\bX_2}{\bY}$.
\end{fact}

\begin{fact}[dependence]
Let $\bX, \bY$ be random variables such that $\bX_2$ is completely determined by $\bX_1$ (i.e. $\bX_2 = f(\bX_1)$ where $f$ is a function). Then we have the following:
\begin{itemize}
    \item $\H\sqbcond{\bX_2}{\bX_1} = 0$;
    \item $\H\sqbcond{\bY}{\bX_1} \leq \H\sqbcond{\bY}{\bX_2}$.
\end{itemize}
\end{fact}

Finally, a crucial notion in our proof is the \emph{mutual information} $\I[\bX:\bY]$ of two random variables $\bX$ and $\bY$, which represents ``how much information $\bX$ reveals about $\bY$'', or symmetrically, ``how much information $\bY$ reveals about $\bX$''.
\begin{definition}[mutual information]
Given two random variables $\bX, \bY$, the \emph{mutual information} of $\bX$ and $\bY$
is the quantity $\I[\bX:\bY] \coloneqq \H[\bX] - \H\sqbcond{\bX}{\bY}$.
\end{definition}

\begin{fact}[symmetry of mutual information]
$\I[\bX:\bY] = \I[\bY:\bX]$.
\end{fact}

\section{The less it is queried, the more it is revealed}
\label{sec:hw-lb}

\paragraph*{Hamming weight: a multi-output function.}
Even though the main function of interest in this paper is the majority function, we will first
prove a lower bound for the related \emph{Hamming weight} function: a ``multi-output'' function (as opposed to binary-output) that reveals the entire Hamming weight of the input string.
\begin{definition}
Let $\HW_n : \zo^n \to \set{0,1,\ldots, n} : x \mapsto |x| = x_1 + \cdots + x_n$ be the Hamming weight function.
\end{definition}

The most natural way to
express $\HW_n$ as $h(g_1, \ldots, g_m)$
is to split the $n$ variables into groups of $k$ variables, and for each group to create $\ceil{\log (k+1)}$ inner functions $g_j$, each computing one bit of the sum of the $k$ variables in this group.
The function $h$ can then compute the Hamming weight of the whole input string by first recovering the sum for each group, then adding them up.
Clearly, each of the inner functions is $k$-local,
and there are $O(\frac{n}{k}\log k)$ of them, so $\CC_k(\HW_n) \leq O(\frac{n}{k}\log k)$.
We will show that this is optimal: $\CC_k(\HW_n) = \Theta(\frac{n}{k}\log k)$.

\paragraph*{Simple counting does not give much.}
It is easy to see $\CC_k(\HW_n)>n/k$:
suppose that there were a way to represent $\HW_n$ as $h(g_1, \ldots, g_m)$ with $m=n/k$.
Then each of the inner functions $g_j$ would query $k$ variables and each variable would be queried only once.
Take $g_1$, and consider the $k$ variables it queries.
Say that we fix all other variables to $0$.
Then there are still $k+1$ possible values for the Hamming weight.
But this fixes the outputs of $g_2,\ldots, g_m$, so there can only be two possible values for $h(g_1, \ldots, g_m)$ (one where $g_1$ outputs $0$ and the other where $g_1$ outputs $1$), so we have a contradiction.

In general, it is easy to see that each set of $r$ variables must collectively be queried by at least $\log(r+1)$ different inner functions.
But this observation is not enough to show that the average variable will need to be queried a super-constant number of times.

\paragraph*{A very counterintuitive lemma.}
Basic counting arguments like the above do not seem to say anything that keeps a variable from being queried by only a constant number of inner functions.
But it turns out there \emph{is} something we can say about variables that are queried by few inner functions: the outputs of the inner functions must reveal a lot about their value, in terms of mutual information.
Put another way, this means that if you are given the outputs of the inner functions, you can often guess the value of such variables better than random chance.

\begin{lemma}[key lemma]
\label{lemma:main}
Suppose that $\HW_n = h(g_1, \ldots, g_m)$ and that variable $i$ is queried at most $q$ of the inner functions $g_1, \ldots, g_m$.
Let $\bX \sim \{0,1\}^n$ be a uniformly random input.
Then $\I\sqb{\bX_i:g_1(\bX), \ldots, g_m(\bX)} \geq 2^{-O(q)}$.
\end{lemma}

Consider how counterintuitive this statement is: it says that the \emph{less} the $i\th$ variable is queried by the inner functions $g_1, 
\ldots, g_m$, the \emph{more} it will be revealed by their collective outputs.

\paragraph*{A lower bound for Hamming weight, assuming~\Cref{lemma:main}.}
We will not prove \Cref{lemma:main} until \Cref{sec:proof-main}, but let us see how it implies the lower bound we want.
First, we will show the intuitive fact that ``if the inner functions reveal a lot about many of the variables, then there must be many inner functions''.
\begin{corollary}
\label{corollary}
Suppose that $\HW_n = h(g_1, \ldots, g_m)$ and each variable is queried by at most $q$ of the inner functions $g_1, \ldots, g_m$.
Then $m \geq n \cdot 2^{-O(q)}$.
\end{corollary}
\begin{proof}
We bound the quantity $\I\sqb{\bX:g_1(\bX), \ldots, g_m(\bX)}$ in two ways.
First, we give an upper bound on $m$ based on the fact that there are only $m$ inner functions, each of which outputs a single bit.
Then we lower bound this same quantity using~\Cref{lemma:main}.
Combining these bounds yields the corollary.

On the one hand,
\begin{align*}
\I\sqb{\bX:g_1(\bX), \ldots, g_m(\bX)}
&= \H\sqb{g_1(\bX), \ldots, g_m(\bX)} - \H\sqbcond{g_1(\bX), \ldots, g_m(\bX)}{\bX}\\
&= \H[g_1(\bX), \ldots, g_m(\bX)]\tag{because $\bX$ determines $g_j(\bX)$ completely}\\
&\leq \H[g_1(\bX)] + \ldots + \H[g_m(\bX)]\tag{because $\H[\cdot]$ is subadditive}\\
&\leq m.\tag{because each $g_j$ has a binary output}
\end{align*}
On the other hand,
\begin{align*}
\I\sqb{\bX:g_1(\bX), \ldots, g_m(\bX)}
&= \H\sqb{\bX} - \H\sqbcond{\bX}{g_1(\bX), \ldots, g_m(\bX)}\\
&= \p*{\sum_{i\in[n]}\H\sqb{\bX_i}} - \H\sqbcond{\bX}{g_1(\bX), \ldots, g_m(\bX)}\tag{because the $\bX_i$ are independent}\\
&\geq \sum_{i\in[n]}\p*{\H\sqb{\bX_i} - \H\sqbcond{\bX_i}{g_1(\bX), \ldots, g_m(\bX)}}\tag{because $\H\sqbcond{\cdot}{\bY}$ is subadditive for any $\bY$}\\
&= \sum_{i\in[n]}\I\sqb{\bX_i:g_1(\bX), \ldots, g_m(\bX)}\\
&\geq \sum_{i\in[n]}2^{-O(q)}\tag{by \Cref{lemma:main}}\\
&= n \cdot 2^{-O(q)}.
\end{align*}
The corollary follows from combining these two bounds.
\end{proof}

We can now use \Cref{corollary} to deduce the desired lower bound.
\begin{theorem}[composition complexity of $\HW_n$]
\label{thm:hamming-weight}
$\CC_k(\HW_n) \geq \Omega(\frac{n}{k}\log k)$.
\end{theorem}
\begin{proof}
Suppose that $\HW_n = f(g_1, \ldots, g_m)$, where each of the inner functions $g_1, \ldots, g_m$ is $k$-local.
By \Cref{cor:wlog}, we only have to show that $m \geq \Omega(\frac{n}{k} \log k)$ under the assumption that each variable is queried at most $q \coloneqq \frac{mk}{n}$ times.
Then by \Cref{corollary},
\[
\frac{nq}{k}
= m
\geq n \cdot 2^{-O(q)}.
\]
Rearranging, we get
\[q \cdot 2^{O(q)} \geq k \Rightarrow q \geq \Omega(\log k),\]
which means $m = \frac{n}{k}q \geq \Omega(\frac{n}{k}\log k)$.
\end{proof}

\section{From Hamming weight to majority}
\label{sec:hw-to-maj}

In the previous section, we gave a lower bound on the composition complexity for the Hamming weight function $\HW_n$.
This is a multi-output function (it has more than two possible outputs), and the proof arguably relied a lot on this fact.
In this section, we show how to extend the lower bound to $\Maj_n$, a binary-output function, by manipulating it to become multi-output.

\begin{definition}[majority function]
Let
\[
\Maj_n :\zo^n \to \zo: x \mapsto
\begin{cases}
1 & \text{if $|x| \geq n/2$}\\
0 & \text{otherwise}
\end{cases}
\]
be the majority function.
\end{definition}

\subsection{Reduction toolkit}
Before proving the lower bound, let us introduce the tools we will use to reduce majority to Hamming weight.

\paragraph*{Control variables.}
$\Maj_n(x)$ tells us whether $x$ has Hamming weight $\geq n/2$ or $< n/2$, but nothing more.
For example, what if we wanted it to also tell us something special when the Hamming weight is \emph{exactly} $n/2$?
Suppose $n$ is even, then we could figure this out by looking at the values of $\Maj_{n+1}(x_1, \ldots, x_n, 0)$ and $\Maj_{n+1}(x_1, \ldots, x_n, 1)$: indeed,
\begin{itemize}
    \item if $|x|< n/2$ then both return $0$;
    \item if $|x|>n/2$ then both return $1$;
    \item if $|x|=n/2$ then $\Maj_{n+1}(x_1, \ldots, x_n, 0)=0$ but $\Maj_{n+1}(x_1, \ldots, x_n, 1)=1$.
\end{itemize}
Note how the last variable of $\Maj_{n+1}$ does not stay free, but rather is assigned a fixed value that modifies the behavior of $\Maj_{n+1}$ as a function of $x_1,\ldots, x_n$.
For this reason, we will call it a \emph{control variable}, and we will call the $n$ other variables \emph{free variables}.

Now, assume that $\Maj_{n+1}$ can be computed as $h(g_1, \ldots, g_m)$.
Then we can compute the function
\[
f: \zo^n \to \set{0,1,2}: x \mapsto
\begin{cases}
0 & \text{if $|x|<n/2$}\\
1 & \text{if $|x|=n/2$}\\
2 & \text{if $|x|>n/2$}
\end{cases}
\]
as $h'(g^0_1, \ldots, g^0_l, g^1_1, \ldots, g^1_l)$ where
\[
g^0_j(x) \coloneqq g_j(x_1, \ldots, x_n, 0) \qquad g^1_j(x) \coloneqq g_j(x_1, \ldots, x_n, 1).
\]
This suggests a potential approach for proving a lower bound on the composition complexity $\Maj_n$: add enough control variables to it so that the values we get are enough to determine the Hamming weight of $x$, then use the lower bound for $\HW_n$ as a black box.
Indeed, if we know all of the values
\[
\left\{
\begin{array}{c}
     \Maj_{2n-1}(x_1, \ldots, x_n, 0, 0, \ldots, 0) \\
     \Maj_{2n-1}(x_1, \ldots, x_n, 1, 0, \ldots, 0) \\
     \vdots\\
     \Maj_{2n-1}(x_1, \ldots, x_n, 1, 1, \ldots, 1)
\end{array}
\right.
\]
then we can recover $|x|$.
However, this would cause a huge blowup in the number of inner functions needed to compute it:
assuming $\Maj_{2n-1}$ can be computed as $h(g_1, \ldots, g_m)$, this would only guarantee that $\HW_n$ can be computed as a function of the $nm$ components 
\[g_j(x_1, \ldots, x_n, 0, \ldots, 0), \ldots, g_j(x_1, \ldots, x_n, 1, \ldots, 1)\]
for $j \in [m]$, thus the $\Omega(\frac{n}{k}\log k)$ lower bound we have for $\HW_n$ would fail to give any nontrivial lower bound on $m$.
So we need to do this in a smarter way.

\paragraph*{Buffer variables.}
We will avoid this blowup by using some \emph{buffer variables} to ``isolate'' the control variables from the free variables.
Assume that $\Maj_n = h(g_1, \ldots, g_m)$.
Suppose we split the variables into three sets $\Ifree, \Icontrol, \Ibuffer \subseteq [n]$ such that none of the inner functions $g_j$ queries variables from both $\Ifree$ and $\Icontrol$ (informally, $\Ibuffer$ acts as a buffer between $\Ifree$ and $\Icontrol$; see \Cref{fig:allowed-gj}).

\input{fig-allowed-gj}

Then if we fix the value of the variables in $\Ibuffer$, we can avoid a blow up in the number of times each variable is queried.
Let us informally see why.
There are two types of inner functions $g_j$:
\begin{itemize}
    \item if $g_j$ queries a free variable, then it does not query any control variables, so no blowup will happen;
    \item if $g_j$ does \emph{not} query any free variables, then it only depends on the control variables and the buffer variables, both of which are known, so in that case we can recompute the output of $g_j$ ourselves.
\end{itemize}
A bit more formally, let $x^{(\Icontrol \mapsto 0)}$ denote $x$ with all variables in $\Icontrol$ set to $0$.
Then we can recover $\Maj_n(x)$ if we know only $g_1(x^{(\Icontrol \mapsto 0)}), \ldots, g_m(x^{(\Icontrol \mapsto 0)})$ as well as the value of the variables in $\Icontrol$ and $\Ibuffer$.
Since $g_j(x^{(\Icontrol \mapsto 0)})$ does not depend on the value of the variables in $\Icontrol$, no blowup happens when we start to vary their values.

\paragraph*{Partial functions.}
Even with the help of buffer variables, though, it turns out we will not be able to have enough control variables to compute $\HW_n$ from $\Maj_n$.\footnote{Or more precisely, to compute $\HW_{n'}$ from $\Maj_n$ for some $n' = \Omega(n)$.}
We will only be able to compute the Hamming weight correctly on a fraction of the possible inputs.
Therefore we need to adapt our techniques \blue{from \Cref{sec:hw-lb}} to a ``partial functions'' setting, where $f(x) = |x|$ is only guaranteed when $x$ is in a subset $D \subseteq \zo^n$ of the possible inputs.

First, the following lemma is a generalization of our key lemma \Cref{lemma:main}, which gave a lower bound on the mutual information between $\bX_i$ and the inner function outputs $g_1(\bX), \ldots, g_m(\bX)$ when variable $i$ is queried few times.
What differs from \Cref{lemma:main} is that the result gives a nontrivial bound only when the probability $\Pr[\bX^{\oplus i} \not\in D]$
is small.
This probability measures how likely you are to leave $D$ if you start out in $D$ then flip the $i\th$ bit, so intuitively, the bigger $D$ is, the smaller it will be.
\begin{restatable}[key lemma, ``partial functions'' version]{lemma}{keylemma}
\label{lemma:main-partial}
\label{LEMMA:MAIN-PARTIAL}
Let $D \subseteq \zo^n$ be a non-empty subset of the hypercube (the ``domain'' on which $f$ computes the Hamming weight), and let
$f: \zo^n \to \set{0, \ldots, n}$ be a function such that $\forall x \in D$, $f(x) = |x|$.
Suppose that $f = h(g_1, \ldots, g_m)$ and that variable $i$ is queried \blue{by} at most $q$ of the inner functions $g_1, \ldots, g_m$.
Let $\bX \sim D$ be a random input drawn uniformly over $D$.
Then $\H\sqbcond{\bX_i}{g_1(\bX), \ldots, g_m(\bX)} \leq 1 + \Pr[\bX^{\oplus i} \not\in D] - 2^{-O(q)}$.
\end{restatable}
\blue{We postpone the proof of this lemma until \Cref{sec:proof-main}. For now,} let us see why it implies \Cref{lemma:main}.
\begin{proof}[Proof of \Cref{lemma:main} assuming \Cref{lemma:main-partial}]
Set $f \coloneqq \HW_n$.
Then $f(x) = |x|$ for all $x \in \zo^n$, so we can set $D \coloneqq \zo^n$.
This means that $\bX$ is a uniformly random input, and $\Pr[\bX^{\oplus i} \not\in D] = 0$.
Therefore, after rearranging, we obtain
\begin{align*}
    2^{-O(q)}
    &\leq 1 - \H\sqbcond{\bX_i}{g_1(\bX), \ldots, g_m(\bX)}\\
    &= \H\sqb{\bX_i} - \H\sqbcond{\bX_i}{g_1(\bX), \ldots, g_m(\bX)}\tag{because $\bX$ is uniformly random}\\
    &= \I\sqb{\bX_i : g_1(\bX), \ldots, g_m(\bX)}.\tag{by definition of mutual information}
\end{align*}
\end{proof}

We can also prove an analog of \Cref{corollary}, which gave a lower bound on the number of inner functions $m$ assuming that many variables are queried few times by the inner functions.
\begin{corollary}
\label{corollary:partial}
There is an integer constant $C>0$ such that the following holds.
Let $D \subseteq \zo^n$ be a non-empty subset of the hypercube, and let
$f: \zo^n \to \set{0, \ldots, n}$ be a function such that $\forall x \in D$, $f(x) = |x|$.
Suppose that $f = h(g_1, \ldots, g_m)$, that each variable is queried by at most $q>0$ of the inner functions $g_1, \ldots, g_m$, and that $\Pr_{\bX \sim D}[\bX^{\oplus i} \not \in D] \leq 2^{-Cq}$.
Then $m + (n-\log|D|) \geq n \cdot 2^{-O(q)}$.
\end{corollary}

\begin{proof}
Let $C'$ be a possible integer value for the constant in the $O(\cdot)$ of \Cref{lemma:main-partial}.
We will prove the corollary for $C \coloneqq C'+1$.
The proof works in the same way as the proof of \Cref{corollary}: we will upper and lower bound the quantity $\I\sqb{\bX:g_1(\bX),\ldots, g_m(\bX)}$.
On the one hand,
\begin{align*}
\I\sqb{\bX:g_1(\bX), \ldots, g_m(\bX)}
&= \H\sqb{g_1(\bX), \ldots, g_m(\bX)} - \H\sqbcond{g_1(\bX), \ldots, g_m(\bX)}{\bX}\\
&= \H[g_1(\bX), \ldots, g_m(\bX)]\tag{because $\bX$ determines $g_j(\bX)$ completely}\\
&\leq \H[g_1(\bX)] + \ldots + \H[g_m(\bX)]\tag{because $\H[\cdot]$ is subadditive}\\
&\leq m.\tag{because each $g_j$ has a binary output}
\end{align*}
On the other hand,
\begin{align*}
\I\sqb{\bX:g_1(\bX), \ldots, g_m(\bX)}
&= \H\sqb{\bX} - \H\sqbcond{\bX}{g_1(\bX), \ldots, g_m(\bX)}\\
&\geq \H\sqb{\bX} - \sum_{i\in[n]}\H\sqbcond{\bX_i}{g_1(\bX), \ldots, g_m(\bX)}\tag{because $\H\sqbcond{\cdot}{\bY}$ is subadditive for any $\bY$}\\
&\geq \H\sqb{\bX} - \sum_{i\in[n]}\p*{1+\Pr[\bX^{\oplus i} \not\in D] - 2^{-C'q}}\tag{by \Cref{lemma:main-partial}}\\
&\geq \H\sqb{\bX} - n + \sum_{i\in[n]}\p*{2^{-C'q} - 2^{-(C'+1)q}}\tag{by assumption that $\Pr_{\bX \sim D}[\bX^{\oplus i} \not \in D] \leq 2^{-Cq}$}\\
&= \log|D| - n + \sum_{i\in[n]}2^{-O(q)}.\tag{because $\bX$ is uniform on $D$, and because $q>0$}
\end{align*}
The corollary follows from combining these two bounds.
\end{proof}

\subsection{A lower bound for majority}
With all the tools in hand, let us prove a lower bound on the composition for $\Maj_n$ by reducing it to a partial version of Hamming weight.
\begin{theorem}[composition complexity of $\Maj_n$]
$\CC_k(\Maj_n) \geq \Omega(\frac{n}{k} \cdot \min(\log k, \log \frac{n}{k}))$.
\end{theorem}
\begin{proof}
Suppose that $\Maj_n = h(g_1, \ldots, g_m)$ where each of the inner functions $g_j$ is $k$-local.
By \Cref{cor:wlog}, it is enough to show that $m \geq \Omega(\frac{n}{k} \cdot \min(\log k, \log \frac{n}{k}))$ under the assumption that each variable is queried at most $q \coloneqq \frac{mk}{n}$ times.
To do this, we will show that $q \geq \Omega(\min(\log k, \log \frac{n}{k}))$.

The plan of the proof is as follows.
We will start by defining the sets of variables $\Ifree$, $\Ibuffer$ and $\Icontrol$.
We will make sure that $|\Ifree| = \Omega(n)$, and we will play with the variables in $\Icontrol$ to show that the inner functions $g_1, \ldots, g_m$ must compute the Hamming weight over a large subset $D$ of $\zo^{\Ifree}$.
From there, we will apply \Cref{corollary:partial}, and only asymptotic calculations will remain, similar to the proof of \Cref{thm:hamming-weight}.

\paragraph*{Defining $\Ifree$, $\Ibuffer$ and $\Icontrol$.}
Let $C$ be the constant in \Cref{corollary:partial}.
Assume that $2^{2Cq+1} qk \leq n/2$ (otherwise, $q = \Omega(\log \frac{n}{k})$ and we are done).
Let $\Icontrol$ be any $2^{2Cq+1}$-element subset of $[n]$, let $J \subseteq [m]$ be the set of inner functions that query some variable in $\Icontrol$, and let $I'\subseteq [n]$ be the set of all variables queried by at least one of the inner functions in $J$.
That is, let $I'$ be the set of variables that are queried by some inner function $g_j$ that also queries a variable in $\Icontrol$.

Since each variable in $\Icontrol$ is queried by at most $q$ inner functions, and each inner function queries at most $k$ variables, we have $|I'| \leq |I| \cdot q \cdot k = 2^{2Cq+1} qk \leq n/2$.
In case $|I'| < 2|\Icontrol|+1$, extend $I'$ to include some more variables of $[n]$ until it reaches size $2|\Icontrol|+1$.
After this, we still have $|I'| \leq \max(n/2, 2|\Icontrol|+1) = \max(n/2, 2\cdot 2^{2Cq+1}+1) \leq n/2$.

\input{fig-free-buffer-control}

Let $\Ifree \coloneqq [n] \setminus I'$, and let $\Ibuffer \coloneqq  I' \setminus \Icontrol$ (see \Cref{fig:free-buffer-control}).
Note that $\Ifree \cup \Ibuffer \cup \Icontrol = [n]$.
Since $\Ifree$ contains no element of $I'$, it is clear that no inner function queries both variables in $\Ifree$ and $\Icontrol$.
Also, defining $\nfree \coloneqq |\Ifree|$, we have
\begin{equation}
    \label{eq:nfree}
    \nfree = n - |I'| \geq n - n/2 = n/2.
\end{equation}

\paragraph*{Computing a partial version of $\HW_{\nfree}$ from $\Maj_n$.}
We will split the input into three parts $x = \xfree \circ \xbuffer \circ \xcontrol$, where $\xfree \in \zo^{\Ifree}$, $\xbuffer \in \zo^{\Ibuffer}$ and $\xcontrol \in \zo^{\Icontrol}$.
First, let us fix a value for $\xbuffer$.
We will fix $\xbuffer$ to be some vector with exactly $|\xbuffer| = \ceil{n/2}-\ceil{\nfree/2} - 2^{2Cq}$ ones.
To make sure this is possible, we need to show $0 \leq \ceil{n/2}-\ceil{\nfree/2} - 2^{2Cq} \leq |\Ibuffer|$.
Firstly, we have that
$n-\nfree = |\Ibuffer| + |\Icontrol| \geq |\Icontrol| = 2^{2Cq+1}$,
so $\ceil{n/2}-\ceil{\nfree/2} \geq \floor*{\frac{n - \nfree}{2}} \geq \floor*{\frac{2^{2Cq+1}}{2}} = 2^{2Cq}$, which gives the lower bound.
Secondly, we have $|\Ibuffer| \geq |I'| - |\Icontrol| \geq |\Icontrol|+1$ (since we made sure that $|I'| \geq 2|\Icontrol|+1$) and thus
\[
|\Ibuffer| \geq \frac{|\Ibuffer|+|\Icontrol|+1}{2} = \frac{n-|\Ifree|+1}{2} = \frac{n-\nfree+1}{2} \geq \ceil{n/2} - \ceil{\nfree/2},
\]
which gives the upper bound.

Since each inner function $g_j(x)$ is either completely determined by $\xfree$ and $\xbuffer$ or completely determined by $\xbuffer$ and $\xcontrol$, in order to compute $\Maj_n(x)$, it is enough to know $g_j(\xfree \circ \xbuffer \circ 0^{\Icontrol})$ for all $j \in [m]$ (where $0^{\Icontrol} \in \zo^{\Icontrol}$ is the all zeros string) as well as the value of $\xbuffer$ and $\xcontrol$.
So even if $g_j(\xfree \circ \xbuffer \circ 0^{\Icontrol})$ for $j \in [m]$ is all we know about $\xfree$, we can compute the value of $\Maj_n(\xfree \circ \xbuffer \circ \xcontrol)$ for any $\xcontrol$ we desire.
Thus, setting $|\xcontrol|$ between $0$ and $|\Icontrol|=2^{2Cq+1}$, we can figure out:
\begin{itemize}
    \item whether $|\xfree| + |\xbuffer| + 0 \geq \ceil{n/2}$;
    \item whether $|\xfree| + |\xbuffer| + 1 \geq \ceil{n/2}$;
    \item \ldots
    \item whether $|\xfree| + |\xbuffer| + 2^{2Cq+1} \geq \ceil{n/2}$.
\end{itemize}
Given that we set $|\xbuffer| = \ceil{n/2} - \ceil{\nfree/2} - 2^{2Cq}$, this means we can distinguish between the following cases:
\begin{itemize}
    \item $|\xfree| \geq \ceil{\nfree/2} + 2^{2Cq}$;
    \item $|\xfree| = \ceil{\nfree/2}+2^{2Cq}-1$;
    \item \ldots
    \item $|\xfree| = \ceil{\nfree/2}-2^{2Cq}$;
    \item $|\xfree| < \ceil{\nfree/2}-2^{2Cq}$.
\end{itemize}
Therefore, we can form a function $f(\xfree) = h'(g_1(\xfree \circ \xbuffer \circ 0^{\Icontrol}), \ldots, g_m(\xfree \circ \xbuffer \circ 0^{\Icontrol}))$ such that
\[
f(\xfree) = 
\begin{cases}
    \ceil{\nfree/2} + 2^{2Cq} & \text{if $|\xfree| \geq \ceil{\nfree/2} + 2^{2Cq}$}\\
    \ceil{\nfree/2} + 2^{2Cq}-1 & \text{if $|\xfree| = \ceil{\nfree/2} + 2^{2Cq}-1$}\\
    \ldots &\\
    \ceil{\nfree/2}-2^{2Cq} & \text{if $|\xfree| = \ceil{\nfree/2}-2^{2Cq}$}\\
    \ceil{\nfree/2}-2^{2Cq}-1 & \text{if $|\xfree| < \ceil{\nfree/2}-2^{2Cq}$.}
\end{cases}
\]

\paragraph*{Applying \Cref{corollary:partial}.}
This function $f$ computes the Hamming weight correctly on the set
\[
D = \setcond*{\xfree \in \zo^{\Ifree}}{\ceil{\nfree/2}-2^{2Cq}-1 \leq |\xfree| \leq \ceil{\nfree/2} + 2^{2Cq}}.
\]
We now prepare to apply \Cref{corollary:partial} to $f$.
Clearly, each input variable of $f$ is queried by at most $q$
of the inner functions $g_1(\xfree \circ \xbuffer \circ 0^{\Icontrol}), \ldots, g_m(\xfree \circ \xbuffer \circ 0^{\Icontrol})$.
Let us see check that $D$ satisfies the condition of \Cref{corollary:partial}.
Observe that $|\bX^{\oplus i}| = |\bX|\pm 1$, so when $\ceil{\nfree/2}-2^{2Cq} \leq |\bX| < \ceil{\nfree/2}+2^{2Cq}$, we have $\bX^{\oplus i} \in D$. Therefore,
\begin{align*}
    \Pr_{\bX \sim D}[\bX^{\oplus i} \not\in D]
    &\leq \Pr_{\bX \sim D}[|\bX| = \ceil{\nfree/2} - 2^{2Cq} - 1 \vee |\bX| = \ceil{\nfree/2} + 2^{2Cq}]\\%\tag{because when , $\bX^{\oplus i}$ is definitely in $D$}\\
    &\leq \frac{2}{\#%
    \set{\ceil{\nfree/2}-2^{2Cq}-1, \ldots, \ceil{\nfree/2} + 2^{2Cq}}}\tag{because the further $|\bX|$ is from $\nfree/2$, the fewer possible values there are for $\bX$}\\
    &= \frac{2}{2^{2Cq+1}+2}\\
    &\leq 2^{-2Cq}.
\end{align*}
Thus we can apply \Cref{corollary:partial} and obtain
\begin{equation}
    \label{eq:from-corollary}
    m + (\nfree - \log |D|) \geq \nfree \cdot 2^{-O(q)}.
\end{equation}

Now, we also have
\begin{align*}
    \log|D|
    &\geq \log\abs*{\setcond{\xfree \in \zo^{\Ifree}}{|\xfree| = \ceil{\nfree/2}}}\nonumber\\
    &\geq \log(\Omega(2^{\nfree}/\sqrt{\nfree}))\nonumber\\
    &= \nfree - O(\log \nfree),
\end{align*}
so $\nfree - \log |D| = O(\log \nfree) = O(\log n)$.
Plugging this into \eqref{eq:from-corollary}, we get
\begin{equation}
    \label{eq:before-asymptotics}
    m + O(\log n) \geq \nfree \cdot 2^{-O(q)}.
\end{equation}

\paragraph*{Asymptotics and conclusion.}
Since each inner function $g_1, \ldots, g_m$ queries $k$ variables and each variable is queried at most $q$ times, we have
$mk \leq nq$. %
Therefore,
\begin{align*}
    \frac{nq}{k} + O(\log n)
    &\geq m + O(\log n)\\%\tag{by \eqref{eq:m-vs-q} and \eqref{eq:nfree-logd}}\\
    &\geq \nfree \cdot 2^{-O(q)}\tag{by \eqref{eq:before-asymptotics}}\\
    &\geq \frac{n}{2} \cdot 2^{-O(q)}.\tag{by \eqref{eq:nfree}}
\end{align*}
Rearranging, we get $2^{O(q)} \p*{\frac{q}{k} + \frac{O(\log n)}{n}} \geq 1/2$, so either
\[\frac{2^{O(q)}q}{k} \geq 1/4 \Rightarrow q = \Omega(\log k)\]
or
\[\frac{2^{O(q)}O(\log n)}{n} \geq 1/4 \Rightarrow q = \Omega(\log n) \geq \Omega(\log k).\qedhere\]
\end{proof}

\section{Proof of~\Cref{lemma:main-partial}}
\label{sec:proof-main}
In this section we present the proof of our key lemma \Cref{lemma:main-partial}, which is a generalization of \Cref{lemma:main} to functions that only compute the Hamming weight correctly on a subset of the inputs.

\keylemma*

We will start by proving a special case to build intuition, then prove the full version using a similar approach.

\subsection{Warmup: $q=1$ for total functions}
In preparation for proving \Cref{lemma:main-partial} in its full generality, let us prove the special case $q=1$ of \Cref{lemma:main} (the ``total functions'' version), which is simpler and illustrates the core idea quite well.
\begin{proposition}[``baby version'' of \Cref{lemma:main}]
Suppose that $\HW_n = h(g_1, \ldots, g_m)$ and that variable $i$ is queried by only one of the inner functions $g_1, \ldots, g_m$.
Let $\bX \sim \{0,1\}^n$ be a uniformly random input.
Then $\I\sqb{\bX_i:g_1(\bX), \ldots, g_m(\bX)} = 1$.
That is, the outputs of the inner functions $g_1(\bX), \ldots, g_m(\bX)$ determine $\bX_i$ completely.
\end{proposition}
\begin{proof}
Say without loss of generality that the only inner function which queries variable $i$ is $g_1$.
Fix any input $x$.
We will show how to recover $x_i$ from the values $g_1(x),\ldots, g_m(x)$.

First, set aside $g_1(x)$, and consider what values $|x|$ could take if we knew only $g_2(x), \ldots, g_m(x)$.
By our assumption $\HW_n = h(g_1, \ldots, g_m)$, $|x| = h(g_1(x), \ldots, g_m(x))$, so $|x|$ must belong to the set
\[
W = \set{h(0, g_2(x), \ldots, g_m(x)), h(1, g_2(x), \ldots, g_m(x))}.
\]

Now, consider the input $x^{\oplus i}$ ($x$ with its $i\th$ coordinate flipped).
Since only $g_1$ depends on the $i\th$ coordinate, $x^{\oplus i}$ must share the same output values for $g_2, \ldots g_m$.
This means that \emph{both} Hamming weights $|x|$ and $|x^{\oplus i}|$ must belong to the set $W$.
But we know that $|x^{\oplus i}|$ is equal to either $|x| + 1$ (when $x_i = 0$) or $|x| - 1$ (when $x_i = 1$), and $|W| \leq 2$, so that means that $W$ must be of the form
\[
W = \set{|x|, |x^{\oplus i}|} = \set{w, w+1}
\]
for some $w$.

So to find the value of $x_i$, it suffices to check whether $|x| = w$ (in which case $x_i = 0$) or $|x| = w+1$ (in which case $x_i = 1$).
Since both $W$ and $|x| = h(g_1(x), \ldots, g_m(x))$ can easily be computed from $g_1(x), \ldots, g_m(x)$, this gives us a way to recover $x_i$ from $g_1(x), \ldots, g_m(x)$.
\end{proof}

\subsection{General case}
The proof of \Cref{lemma:main-partial} is a generalization of the trick above: we fix the outputs of the inner functions that \emph{do not} query variable $i$,
then consider the (relatively small) set of possible values for $|\bX|$ given those outputs, and finally use our knowledge of $|\bX|$ to guess $\bX_i$ better than random chance.

\begin{proof}[Proof of \Cref{lemma:main-partial}]
First, let us show that the inequality holds when $q = 0$.
If $q=0$, then $f$ does not depend on variable $i$ at all, so for any $x \in \set{0,1}^n$, $f(x) = f(x^{\oplus i})$.
This means we cannot simultaneously have $f(x) = |x|$ and $f(x^{\oplus i}) = |x^{\oplus i}|$.
Thus whenever $x \in D$, $x^{\oplus i}$ cannot also be in $D$.
As a result, $\Pr[\bX^{\oplus i} \not\in D] = 1$, so the right hand side becomes
\[1 + \Pr[\bX^{\oplus i} \not\in D] - 2^{-O(q)} = 1 + 1 - 1 = 1,\]
and the inequality is trivially verified.
Having proved the inequality for $q=0$, we may assume that $q \geq 1$ for the remainder of the proof.

Let $J_i \subseteq [m]$ be the set of inner functions that depend on variable $i$ (so $|J_i| \leq q$), and let $\Jc \coloneqq [m] \setminus J_i$.
Using the fact that $g_1(\bX), \ldots, g_m(\bX)$ determine $|\bX|$ completely, we can weaken the conditioning to get
\[
\H\sqbcond{\bX_i}{g_1(\bX), \ldots, g_m(\bX)}
\leq \H\sqbcond{\bX_i}{\set{g_j(\bX)}_{j \in \Jc}, |\bX|},
\]
and by definition of conditional entropy,
\begin{align*}
&\H\sqbcond{\bX_i}{\set{g_j(\bX)}_{j \in \Jc}, |\bX|}\\
&\qquad= \E_{\bY \sim D}[-\log\Pr\sqbcond{\bX_i=\bY_i}{(\forall j \in \Jc, g_j(\bX)=g_j(\bY))\wedge|\bX|=|\bY|}]\\
&\qquad= \E_{\bY \sim D}[H_2(\Pr\sqbcond{\bX_i=1}{(\forall j \in \Jc, g_j(\bX)=g_j(\bY))\wedge|\bX|=|\bY|})],
\end{align*}
where $H_2(p) \coloneqq -p\log p - (1-p)\log(1-p)$ is the binary entropy function.
Thus, to prove the lemma, it suffices to show that
\begin{equation}
    \label{eq:really-want}
    \E_{\bY \sim D}[H_2(\Pr\sqbcond{\bX_i=1}{(\forall j \in \Jc, g_j(\bX)=g_j(\bY))\wedge|\bX|=|\bY|})] \leq 1+\Pr[\bX^{\oplus i} \not\in D]-2^{-O(q)}.
\end{equation}

\paragraph*{Partitioning $D$ according to the outputs of the inner functions that do not query the variable $i$.} To prove this inequality, let us first split into cases according to the outputs of the inner functions that do not query variable $i$.
Fix the output values $v \in \zo^\Jc$ of the inner functions that do not query variable $i$, and suppose that those output values are achievable by some input in $D$ (that is, there exists some input $x \in D$ such that $\forall j \in \Jc, g_j(x) = v_j$).
Let $S_v \coloneqq \setcond{x \in \set{0,1}^n}{\forall j \in \Jc, g_j(x) = v_j}$ be the set of inputs that produce these inner function outputs.
Note that the sets $D \cap S_v$ form a partition of $D$.
As we will see later, a key property of $S_v$ is that its image under $f$ has cardinality at most $2^{q}$.

Let us call $v$ \emph{good} if $\Pr\sqbcond{\bX^{\oplus i} \not\in D}{\bX \in S_v} \leq 2\cdot \Pr[\bX^{\oplus i} \not\in D]$ (that is, $v$ is good if $\bX^{\oplus i}$ is not much more likely to lie outside of $D$ when $\bX$ is in $S_v$ than on average).
By Markov's inequality, the sets $S_v$ with good $v$ account for most of the probability mass: that is, $\sum_{\text{$v$ good}} \Pr[\bX \in S_v] \geq 1/2$.

Our objective will be to show that for every good $v$,
\begin{equation}
\label{eq:want}
\E_{\bY \sim D}\sqbcond*{H_2(\Pr\sqbcond{\bX_i=1}{\bX \in S_v \wedge|\bX|=|\bY|})}{\bY \in S_v}
\leq 1 + \Pr[\bX^{\oplus i} \not \in D] - 2^{-O(q)}.
\end{equation}
Informally, the task is: for any good $v \in \zo^\Jc$ (representing the outputs of the inner functions that do not query variable $i$), based only on $v$ and the Hamming weight $|\bX|$, guess the value of $\bX_i$ slightly better than random chance.

Let us see why \eqref{eq:want} implies \eqref{eq:really-want}:
\begin{align*}
    &\E_{\bY \sim D}[H_2(\Pr\sqbcond{\bX_i=1}{(\forall j \in \Jc, g_j(\bX)=g_j(\bY))\wedge|\bX|=|\bY|})]\\
    &\qquad= \sum_v \Pr[\bX \in S_v] \cdot \E_{\bY \sim D}\sqbcond*{H_2(\Pr\sqbcond{\bX_i=1}{\bX \in S_v \wedge|\bX|=|\bY|})}{\bY \in S_v}\tag{by the law of total expectation}\\
    &\qquad\leq \sum_v \Pr[\bX \in S_v] \cdot
    \begin{cases}
    1 + \Pr[\bX^{\oplus i} \not \in D] - 2^{-O(q)} & \text{if $v$ is good}\\
    1 & \text{otherwise}
    \end{cases}\tag{by \eqref{eq:want} and because $H_2(\cdot) \leq 1$}\\
    &\qquad= 1 + \p*{\sum_\text{$v$ good} \Pr[\bX \in S_v]} \p*{\Pr[\bX^{\oplus i} \not \in D] - 2^{-O(q)}}\\
    &\qquad\leq 1 + \Pr[\bX^{\oplus i} \not \in D] - \p*{\sum_\text{$v$ good} \Pr[\bX \in S_v]}2^{-O(q)}\tag{because $\sum_\text{$v$ good} \Pr[\bX \in S_v] \leq 1$}\\
    &\qquad\leq 1 + \Pr[\bX^{\oplus i} \not \in D] - 2^{-O(q)}.\tag{because $\sum_\text{$v$ good} \Pr[\bX \in S_v] \geq 1/2$ and $q > 0$}
\end{align*}

\paragraph*{Proof overview for~\eqref{eq:want}.} 
Fix some good $v \in \zo^\Jc$, and let $W_v \coloneqq \setcond{|x|}{x \in D \cap S_v}$ be the set of possible Hamming weights for inputs in $D \cap S_v$.
When $x \in D$, we know that $|x| = f(x) = h(g_1(x), \ldots, g_m(x))$, and when in addition $x \in S_v$, then for all $j \in \Jc$, the output value $g_j(x)$ is already fixed to $v_j$.
Therefore, when $x \in D \cap S_v$, the Hamming weight $|x|$ can only depend on the remaining $q$ output values $g_j(x)$ for $j \in J_i$, so there are at most $2^{q}$ possible values for $|x|$.
In other words, $|W_v| \leq 2^{q}$.

Informally, using the fact that $W_v$ is small, we will show that there exists a Hamming weight $\myw \in W_v$ that occurs with significant probability within $S_v$, and such that conditioned on $|\bX|=\myw$, the probability that $\bX_i=1$ is not too close to $1/2$.
This means that when $|\bX|=\myw$ we can guess $\bX_i$ slightly better than random chance, and even if we just guess randomly when $|\bX| \neq \myw$, we will have achieved better-than-random-chance accuracy overall.
Formally, we will show that $\exists \myw \in W_v$ such that
\begin{equation}
\label{eq:significant-prob}
\Pr\sqbcond{|\bX|=\myw}{x \in S_v} \geq 2^{-O(q)}
\end{equation}
and
\begin{equation}
\label{eq:away-from-half}
\Pr\sqbcond{\bX_i=1}{\bX \in S_v \wedge |\bX|=\myw} \leq \frac{1}{2} - 2^{-O(q)}.
\end{equation}
Let us see why \eqref{eq:significant-prob} and~\eqref{eq:away-from-half} together imply \eqref{eq:want}:
\begin{align*}
&\E_{\bY \sim D}\sqbcond*{H_2(\Pr\sqbcond{\bX_i=1}{\bX \in S_v \wedge|\bX|=|\bY|})}{\bY \in S_v}\\
&\qquad=\sum_{w \in W_v}\Pr\sqbcond{|\bX|=w}{\bX \in S_v}\\
&\qquad\qquad\qquad\cdot \E_{\bY \sim D}\sqbcond*{H_2(\Pr\sqbcond{\bX_i=1}{\bX \in S_v \wedge|\bX|=|\bY|})}{\bY \in S_v \wedge |\bY| = w}\tag{by the law of total expectation}\\
&\qquad=\sum_{w \in W_v}\Pr\sqbcond{|\bX|=w}{\bX \in S_v}\\
&\qquad\qquad\qquad\cdot \E_{\bY \sim D}\sqbcond*{H_2(\Pr\sqbcond{\bX_i=1}{\bX \in S_v \wedge|\bX|=w})}{\bY \in S_v \wedge |\bY| = w}\tag{replace $|\bY|$ by $w$ using the conditioning}\\
&\qquad=\sum_{w \in W_v}\Pr\sqbcond{|\bX|=w}{\bX \in S_v}\cdot H_2(\Pr\sqbcond{\bX_i=1}{\bX \in S_v \wedge|\bX|=w})\tag{the expectation was constant}\\
&\qquad\leq\sum_{w \in W_v}\Pr\sqbcond{|\bX|=w}{\bX \in S_v}\cdot
\begin{cases}
H_2(\Pr\sqbcond{\bX_i=1}{\bX \in S_v \wedge|\bX|=\myw}) & \text{if $w = \myw$}\\
1 & \text{otherwise}
\end{cases}
\tag{because $H_2(\cdot) \leq 1$}\\
&\qquad= 1 - \Pr\sqbcond{|\bX|=\myw}{\bX \in S_v} \cdot\p*{1-H_2(\Pr\sqbcond{\bX_i=1}{\bX \in S_v \wedge|\bX|=\myw})}\tag{rearrange and use the fact that probabilities sum to 1}\\
&\qquad\leq 1 - 2^{-O(q)} \p*{1-H_2\p*{1/2 - 2^{-O(q)}}}\tag{by \eqref{eq:significant-prob} and \eqref{eq:away-from-half}}\\
&\qquad= 1 - 2^{-O(q)} \p*{1 - \p*{1-2^{-O(q)}}}\tag{because $H_2(1/2-p) = 1-\Omega(p^2)$}\\
&\qquad= 1 - 2^{-O(q)}.
\end{align*}

\paragraph*{Existence of the weight $\myw$.} Let us now find this magical Hamming weight $\myw$.

First of all, assume that $\Pr[\bX^{\oplus i} \not\in D] \leq 2^{-10q}$ (otherwise, we can replace the right-hand side of \eqref{eq:want} by $1$ and the inequality becomes trivial).
Since $v$ is good, this means that $\Pr\sqbcond{\bX^{\oplus i} \not\in D}{\bX \in S_v} \leq 2 \cdot 2^{-10q} = 2^{-10q+1}$.

When $x \in S_v$, $x^{\oplus i}$ must also be in $S_v$, since flipping variable $i$ does not affect the output of any inner function $g_j$ for $j \in \Jc$.
Informally, this means that for most inputs $x \in D \cap S_v$ (those for which $x^{\oplus i} \in D$), we can pair it up with another input $x^{\oplus i}$ also in $D \cap S_v$.
We can identify each such pair by looking at the value of $x^{(i \mapsto 0)}$ ($x$ with its $i\th$ coordinate set to $0$).
Our argument will rely crucially on this pairing of inputs.

For any integer $-1 \leq w \leq n$, let $p_w \coloneqq \Pr\sqbcond{\bX^{\oplus i} \in D \wedge |\bX^{(i \mapsto 0)}| = w}{\bX \in S_v}$ (note that $p_{-1} = 0$ somewhat vacuously).
It is clear that $p_w = 0$ whenever $w,w+1 \not\in W_v$.
This means that the sequence $\set{p_w}_{w\in\set{-1,\ldots, n}}$ has at most $2|W_v| \leq 2^{q+1}$ nonzero elements.
Moreover $\sum_{w =-1}^n p_w = \Pr\sqbcond{\bX^{\oplus i} \in D}{\bX \in S_v} \geq 1- 2^{-10q+1}$, so we must have $\max_{w=-1}^n p_w \geq (1-2^{-10q+1})/2^{q+1} \geq 2^{-q-2}$.
Now, again using the fact that the sequence $\set{p_w}_{w\in\set{-1,\ldots, n}}$ has at most $2^{q+1}$ nonzero elements, to go from $p_{-1}=0$ to this maximum value of $\geq 2^{-q-2}$, the sequence must contain some ``upwards jump'' of at least $2^{-q-2}/2^{q+1} = 2^{-2q-3}$, so there must exist a weight $\myw \in \set{0,\ldots,n}$ such that $p_\myw \geq p_{\myw-1} + 2^{-2q-3}$.

Let us show that $\myw \in W_v$ and that it satisfies \eqref{eq:significant-prob} and \eqref{eq:away-from-half}.
First, note that the set $\setcond{x \in D \cap S_v}{|x| = \myw}$ includes at least the following two disjoint sets:
\begin{enumerate}[(i)]
    \item the inputs $x \in D \cap S_v$ such that $x^{\oplus i} \in D$, $x_i=0$, and $|x^{(i \mapsto 0)}|=\myw$;
    \item the inputs $x \in D \cap S_v$ such that $x^{\oplus i} \in D$, $x_i=1$ and $|x^{(i \mapsto 0)}|=\myw-1$.
\end{enumerate}
It is easy to see that there are exactly $\frac{p_\myw}{2}|D \cap S_v|$ inputs of type (i) and $\frac{p_{\myw-1}}{2}|D \cap S_v|$ inputs of type (ii).
This means that
\begin{align}
    \Pr\sqbcond{|\bX|=\myw}{\bX \in S_v}
    &\geq \frac{p_\myw + p_{\myw-1}}{2}\label{eq:lb-on-prob-w}\\
    &\geq \frac{p_\myw}{2}\tag{because $p_{\myw-1}$ is a probability}\\
    &\geq \frac{2^{-2q-3}}{2}\tag{because $p_{\myw} \geq p_{\myw-1} + 2^{-2q-3} \geq 2^{-2q-3}$}\\
    &= 2^{-O(q)},\nonumber
\end{align}
satisfying \eqref{eq:significant-prob}.
Also, given that $\Pr\sqbcond{|\bX|=\myw}{\bX \in S_v} \geq 2^{-O(q)} > 0$, there must be some $x \in D \cap S_v$ such that $|x| = \myw$, so $\myw \in W_v$.
Finally,
\begin{align*}
&\Pr\sqbcond{\bX_i=1}{\bX \in S_v \wedge |\bX|=\myw}\\
&\qquad= \frac{\Pr\sqbcond{\bX_i = 1 \wedge |\bX| = \myw}{\bX \in S_v}}{\Pr\sqbcond{|\bX| = \myw}{\bX \in S_v}}\tag{by definition of conditional probability}\\
&\qquad= \frac{\Pr\sqbcond{\bX^{\oplus i} \not\in D \wedge \bX_i = 1 \wedge |\bX| = \myw}{\bX \in S_v} + \Pr\sqbcond{\bX^{\oplus i} \in D \wedge \bX_i = 1 \wedge |\bX| = \myw}{\bX \in S_v}}{\Pr\sqbcond{|\bX| = \myw}{\bX \in S_v}}\tag{split according to $\bX^{\oplus i} \overset{?}{\in} D$}\\
&\qquad\leq \frac{\Pr\sqbcond{\bX^{\oplus i} \not\in D}{\bX \in S_v} + \Pr\sqbcond{\bX^{\oplus i} \in D \wedge \bX_i = 1 \wedge |\bX^{(i \mapsto 0)}| = \myw-1}{\bX \in S_v}}{\Pr\sqbcond{|\bX| = \myw}{\bX \in S_v}}\tag{logical consequences}\\
&\qquad\leq \frac{2^{-10q+1} + \frac{p_{\myw-1}}{2}}{\frac{p_\myw + p_{\myw-1}}{2}}\tag{because there are $\frac{p_{\myw-1}}{2}|D \cap S_v|$ inputs of type (ii), and by \eqref{eq:lb-on-prob-w}}\\
&\qquad= \frac{2^{-10q+2} + p_{\myw-1}}{p_{\myw}+p_{\myw-1}}\\
&\qquad\leq \frac{2^{-10q+2} + (p_\myw - 2^{-2q-3})}{p_\myw + (p_\myw-2^{-2q-3})}\tag{because $p_\myw \geq 2^{-2q-3} > 2^{10q+2}$ and $p_{\myw-1} \leq p_{\myw}-2^{-2q-3}$}\\
&\qquad= \frac{1}{2} - \frac{2^{-2q-3}/2 - 2^{-10q+2}}{2p_\myw-2^{-2q-3}}\\
&\qquad\leq \frac{1}{2} - \frac{2^{-2q-5}}{2p_\myw-2^{-2q-3}}\tag{because $q \geq 1$}\\
&\qquad< \frac{1}{2} - \frac{2^{-2q-5}}{2}\tag{because $p_\myw \leq 1$}\\
&\qquad= \frac{1}{2} - 2^{-O(q)},\tag{because $q \geq 1$}
\end{align*}
thus \eqref{eq:away-from-half} is satisfied, and this concludes the proof.
\end{proof}

\section{Conclusion}\label{sec:conclusion}
In this section we outline several
interesting directions for future work.

\subsection{A plan towards better depth-3 lower bounds}
\label{subsec:depth3-plan}
As we pointed out in the introduction, if it were possible to compute $\Maj_n$ as a fan-in-$O(\sqrt{n})$ function of fan-in-$O(\sqrt{n})$ functions, there would be $2^{O(\sqrt{n})}$-size depth-3 circuits for $\Maj_n$.
By showing that $\CC_{\Theta(\sqrt{n})}(\Maj_n) = \omega(\sqrt{n})$, we ruled out this particular way of obtaining $2^{O(\sqrt{n})}$-size depth-3 circuits for $\Maj_n$.
\Cref{thm:main} is therefore a necessary first step for showing that $\Maj_n$ requires $2^{\omega(\sqrt{n})}$-size depth-3 circuits.

If we keep considering constructions of the form $\Maj_n = h(g_1, \ldots, g_m)$, there are several ways we can make both the outer function $h$ and the inner functions $g_1, \ldots, g_m$ more powerful while still giving depth-3 circuit upper bounds.
For example, one could try to prove that $\Maj_n$ cannot be computed as a fan-in-$O(\sqrt{n})$ function of depth-$O(\sqrt{n})$ decision trees,
or even as a depth-$O(\sqrt{n})$ decision tree of depth-$O(\sqrt{n})$ decision trees.
Indeed, if such a construction were possible, the outer function could be transformed into a size-$2^{O(\sqrt{n})}$ DNF (resp. CNF), while the inner functions and their negations could be transformed into size-$2^{O(\sqrt{n})}$  CNFs (resp. DNFs), which would give a $2^{O(\sqrt{n})}$-size $\Sigma^3$ (resp. $\Pi^3$) circuit for $\Maj_n$.

This motivates the following plan towards proving better depth-3 lower bounds for the majority function:
prove lower bounds against computing $\Maj_n$ as $h(g_1, \ldots, g_m)$, where the inner and outer functions are replaced by increasingly powerful objects.
For now, let us conjecture the following first steps, which could be solved independently.

\begin{conjecture}[making the inner functions more powerful]
\label{conj:inner-power}
$\Maj_n$ cannot be represented as $h(g_1, \ldots, g_m)$, where $m = O(\sqrt{n})$ and each $g_j$ is a depth-$O(\sqrt{n})$ decision tree.
\end{conjecture}

\begin{conjecture}[making the outer function more powerful]
\label{conj:outer-power}
$\Maj_n$ cannot be represented as $h(g_1, \ldots, g_m)$, where $h$ is a depth-$O(\sqrt{n})$ decision tree and each $g_j$ depends on only $O(\sqrt{n})$ variables.
\end{conjecture}

In fact, we do not have to climb very far up the power ladder in order to reach the full power of depth-3 circuits.
Indeed, if we were to replace ``depth-$O(\sqrt{n})$ decision tree'' by ``size-$2^{O(\sqrt{n})}$ DNF'' (resp. CNF) in \Cref{conj:outer-power}, this would already be equivalent to the conjecture that $\Maj_n$ cannot be represented by $\Sigma^3$ (resp. $\Pi^3$) circuits of size $2^{O(\sqrt{n})}$.\footnote{In fact, this would still be true even if we replaced ``each $g_j$ depends only on $O(\sqrt{n})$ variables'' by ``each $g_j$ is an OR of $O(\sqrt{n})$ variables''.}

\subsection{Further applications of our techniques}

\paragraph*{An information theoretic toolbox for circuit lower bounds?} In this paper, we showed (in \Cref{lemma:main} and \Cref{lemma:main-partial}) how we can trace the information flow within a circuit of arbitrary functions to prove tight lower bounds. Could one prove lower bounds against boolean circuits via the same technique? In particular, it would be interesting to prove a statement analogous to \Cref{lemma:main}, but quantifying the mutual information between internal nodes of a boolean circuit. 

One challenge of accomplishing this is in the difference in power between the inner functions and the AND/OR gates of boolean circuits. Since the boolean gates are so much weaker than the inner functions (which compute arbitrary boolean functions), we would need to prove much stronger statements about the information flow in order to get tight lower bounds. Nonetheless, we believe that challenges like \Cref{op:depth3lbs} require techniques that can precisely identify the weaknesses of small-depth circuits, and that information theory is a well-suited tool for this task.

\paragraph*{The usefulness of multi-output functions.} 
Multi-output functions played a key role in two steps of the proof. These steps can be distilled as a general plan for proving lower bounds for single-output functions as follows:

\begin{enumerate}[(1)]
    \item Prove a lower bound for a multi-output function (in our case, $\HW_n$). 
    \item Show that the act of computing a desired single-output function (in our case, $\Maj_n$) essentially entails computing the multi-output function (perhaps on a smaller set of inputs).
\end{enumerate}

We believe that this general approach could be applied to a much wider range of models. For example, consider the problem of improving the best lower bounds for depth-3 circuits. If one could prove a $2^{\omega(\sqrt{n})}$ lower bound against depth-3 circuits of $\HW_n$, then this would \emph{automatically} give us better lower bounds for an explicit function (we would get a lower bound of $2^{\omega(\sqrt{n})}/\log n = 2^{\omega(\sqrt{n})}$ for at least one of the output bits of $\HW_n$). In all likelihood, this could be extended (perhaps by the approach in step (2)) to work for $\Maj_n$, thereby resolving \Cref{op:depth3lbs}. 

We usually think of single-output functions as being components of multi-output functions, but this work shows that multi-output functions can be hidden within single-output functions as well. We hope that this will encourage the incorporation of multi-output functions into the lexicon of hard functions to prove lower bounds for.

\paragraph*{Random linear codes.} As a concrete example of functions that one might prove strong lower bounds for using our techniques, we propose random linear error-correcting codes.

In a celebrated work, Paturi et al.~\cite{PPSZ05} showed that error-correcting codes\footnote{More precisely, functions computing whether their input $x$ belongs to a fixed error-correcting code.} with large enough distance and enough codewords require size-$2^{1.282 \sqrt{n}}$-size $\Sigma^3$ circuits, which remains to date the strongest lower bound proved for any explicit function.
A natural question is, was this technique tight?
Can we obtain $2^{\omega(\sqrt{n})}$ lower bounds for all ``good enough'' error-correcting codes?

Unfortunately, the answer is at least partially negative.
Leffman, Pudlák, and Savický \cite{LPS97} showed that there exist linear codes of large distance which have \emph{sparse} parity check matrices.
More precisely, they show the existence of a linear code $f$ of distance $n^{\Omega(1)}$ whose parity check matrix has $\sqrt{n}$ rows, each of which has only $O(\sqrt{n})$ non-zero entries.
This means that we can write $f=h(g_1, \ldots, g_{\sqrt{n}})$, where each of the inner functions $g_j$ computes a parity check over $O(\sqrt{n})$ variables, and $h$ is the AND function.
Using the notation of this paper, $\CC_{O(\sqrt{n})}(f) \leq \sqrt{n}$, and thus $f$ has size-$2^{O(\sqrt{n})}$ $\Sigma^3$ circuits.

However, this code has small depth-3 circuits precisely because its parity check matrix is sparse.
It is easy to see by a counting argument that most linear codes do \emph{not} have sparse parity check matrices.
So it is natural to wonder whether a random linear code (a code whose parity check matrix is chosen uniformly at random) might require large $\Sigma^3$ circuits.
Concretely, we conjecture the following.

\begin{conjecture}
\label{conj:random-codes}
Let $\ell \ll k \ll n$. Let $H \sim \zo^{\ell \times n}$ be a random matrix of $\ell$ parity checks, and let
\[
f: \zo \to \zo^n : x \mapsto
\begin{cases}
1 & \text{if $Hx = \Vec{0}$}\\
0 & \text{otherwise}
\end{cases}
\]
be the corresponding linear code. Then, with high probability,
\begin{enumerate}[(i)]
    \item $\CC_k(f) \geq \Omega\p*{\frac{n}{k} \cdot \frac{\ell}{\log(n/k)}}$;
    \item $f$ requires $\Sigma^3$ circuits of size $2^{\Omega(\sqrt{n\ell / \log n})}$.
\end{enumerate}
\end{conjecture}

The proposed bound $\Omega(\frac{n}{k} \cdot \frac{\ell}{\log(n/k)})$ in point (i) comes from the observation that one can verify any $\log(n/k)$ parity check using $O(n/k)$ inner functions: first split the variables into $O(n/k)$ sets according to their coefficients in each of the $\log(n/k)$ parity check, then compute the parity of each of those $O(n/k)$ sets.
Point (i) simply conjectures that this observation gives the best upper bound on $\CC_k(f)$.
The proposed bound $2^{\Omega(\sqrt{n\ell / \log n})}$ in point (ii) is obtained by conjecturing that the best $\Sigma^3$ circuit size for $f$ will be obtained as $2^{O(\CC_k(f) + k)}$ through the connection with composition complexity (see \Cref{sssec:link-to-small-depth}), then balancing the sum by setting $k \coloneqq \sqrt{n \ell / \log n}$.

If \Cref{conj:random-codes} is true, it would give lower bounds of the form $2^{n^{1-\eps}}$ against $\Sigma^3$ circuits.
It seems likely to us that, depending on how the proof works, the randomness could then be lifted, and it could be extended to some explicit \emph{pseudorandom} linear code.

We think the techniques in our paper would be particularly well-suited to proving \Cref{conj:random-codes}.
Indeed, just like majority has a natural multi-output analog (the Hamming weight function), $f$ has an even more obvious corresponding multi-output function: the function $\Vec{f}(x) \coloneqq Hx$ which gives the outputs of all the parity checks.
Therefore, it seems plausible that one could first prove a lower bound for $\Vec{f}$ using techniques similar to the ones in \Cref{sec:hw-lb}, then extend it to $f$ using our framework for bootstrapping lower bounds from multi-output functions to binary-output functions, which we presented in \Cref{sec:hw-to-maj}.

%% file: intro.tex
\def\colorful{0}

\ifnum\colorful=1
\newcommand{\violet}[1]{{\color{violet}{#1}}}
\newcommand{\orange}[1]{{\color{orange}{#1}}}
\newcommand{\blue}[1]{{{\color{blue}#1}}}
\newcommand{\red}[1]{{\color{red} {#1}}}
\newcommand{\green}[1]{{\color{green} {#1}}}
\newcommand{\pink}[1]{{\color{pink}{#1}}}
\newcommand{\gray}[1]{{\color{gray}{#1}}}

\fi
\ifnum\colorful=0
\newcommand{\violet}[1]{{{#1}}}
\newcommand{\orange}[1]{{{#1}}}
\newcommand{\blue}[1]{{{#1}}}
\newcommand{\red}[1]{{{#1}}}
\newcommand{\green}[1]{{{#1}}}
\newcommand{\gray}[1]{{{#1}}}

\fi

\def\showgray{0}

\renewcommand{\gray}[1]{\vspace{2\baselineskip}}

A basic theme in computer science is the representation of certain functions as the combination of simpler ones.
Indeed, the field of distributed computing and the widespread principle of divide-and-conquer rely on this property of functions.

In this paper we focus on \emph{locality} as our notion of simplicity: a $k$-local function over $n$ variables is one that depends only on $k \ll n$ input coordinates.
This leads us to the following complexity measure of boolean functions, first
studied
by Hrube\v{s} and Rao~\cite{HR15}, which quantifies how easily they can be represented as the combination of local functions:

\begin{definition}[composition complexity]
\label{def:complexity}
The \emph{$k$-composition complexity} of a function $f$, denoted $\CC_k(f)$, is the minimum $m$ such that $f$ can be expressed as $h(g_1, \ldots, g_m)$, where each of the inner functions
$g_j$
queries only $k$ variables.
\end{definition}

Clearly, any function $f$ that actually depends on all $n$ variables must have $\CC_k(f) \geq n/k$, since every variable must be queried at least once.
The parity function shows that this bound can be tight: $\CC_k(\Parity_n) \leq O(n/k)$, since we can let each $g_j$ compute the parity of a set of
$k$ variables, and let $h$ compute the parity of these parities.
However, some functions inherently require more than $n/k$ inner functions: they incur a \emph{composition overhead}.  This motivates defining the \emph{$k$-composition overhead} of  $f$ to be the ratio $\frac{\CC_k(f)}{n/k}$ between its $k$-composition complexity and the ideal $n/k$.

\subsection{This work}
In this paper we study the composition complexity of the majority function.
Prior to our work, this was perhaps the most basic function whose composition complexity was not well understood.
For an upper bound, it is not hard to see that $\CC_k(\Majority_n) \leq O(\frac{n}{k} \log k)$.
We can split the variables into $n/k$ disjoint sets of size $k$ and devote $O(\log k)$ of the inner functions $g_i$ to computing the Hamming weight of each set.
Then $h$ can determine the overall Hamming weight, and output 1 if and only if it is at least $n/2$.
This uses a total of $m \leq O(\frac{n}{k} \log k)$ inner functions.

As for lower bounds, while it is an easy exercise to improve the trivial lower bound of $\ge n/k$
to $> n/k$
in the case of majority,
even a modest lower bound of $\ge 1.1 \frac{n}{k}$ seems challenging to establish.
Our main result is an asymptotically tight lower bound showing that the construction described above is optimal, and that majority has a composition overhead of $\Theta(\log k)$.

\begin{theorem}[Main theorem]\label{thm:main}
For all $\eps > 0$ and $k \leq n^{1 - \eps}$, $\CC_k(\Majority_n) \geq \Omega(\frac{n}{k} \log k)$.
\end{theorem}

In addition to being a natural complexity measure that is of independent interest, our study of composition complexity is further motivated by its relationships to two important models of computation: bounded-width branching programs and small-depth circuits.
These are two of the most intensively studied models in circuit complexity, and majority plays a starring role in the efforts at lower bounds for both models.
As we now elaborate,~\Cref{thm:main} recovers, as a corollary and via an entirely different proof, the current best lower bounds on the length of bounded-width branching programs for majority \cite{AM86,BPRS90}.
It is also the first step in a plan that we propose for lower bounds against depth-$3$ circuits computing majority, a long-standing open problem that dates back to the 1990s \cite{HJP93}; such lower bounds would represent the first improvements over the state of the art for depth-$3$ circuits in over three decades \cite{Has86}.

\subsection{Motivation and implications}\label{sec:motives}

\subsubsection{Bounded-width branching programs}
There is an easy reduction from branching programs to our model: if function $f$ is computed by a bounded-width branching program of length $L$,
then $\CC_k(f) \leq O(L/k)$.
The reduction works by cutting the branching programs into $\sim L/k$ segments of length at most $k$, then for each segment and for each state $s$ at the start of the segment, use
$O(1)$
inner functions to compute which state one would end up at the end of the segment if one started from $s$.
Since the segments have length at most $k$, each of the inner functions depends on at most $k$ variables.

This means that nontrivial lower bounds on the composition complexity $\CC_k(f)$ for any $k$ directly imply nontrivial lower bounds on the length of bounded-width branching programs: $L \geq \Omega(k \cdot \CC_k(f))$.
In particular, setting $k=\sqrt{n}$, by \Cref{thm:main} we obtain that bounded-width branching programs computing the majority function must have length $L \geq \Omega(\sqrt{n} \cdot \frac{n}{\sqrt{n}} \log \sqrt{n}) = \Omega(n \log n)$, recovering the classic lower bound of Alon and Maass~\cite{AM86} and Babai, Pudl{\'a}k, R{\"o}dl, and Szemer{\'e}di~\cite{BPRS90} as a corollary: 

\begin{theorem}[\cite{AM86,BPRS90}]
Any bounded-width branching program computing $\Maj_n$ must have length $\Omega(n \log n)$.
\end{theorem}

This lower bound remains the current state of the art. Interestingly, our techniques are completely different from the techniques used in \cite{AM86} and \cite{BPRS90}: they first prove a Ramsey-theoretic lemma that identifies two sets of variables that are queried in disjoint segments of the branching program, then conclude by a communication complexity argument between them. 
We elaborate on our techniques in~\Cref{sec:techniques}.

\subsubsection{Lower bounds for small-depth circuits computing majority}
\label{sssec:link-to-small-depth}

We view \Cref{thm:main} in part as an essential and necessary first step towards proving stronger lower bounds against small-depth circuits. In this section, we will outline why such lower bounds are particularly interesting, and the relationship between our composition complexity lower bounds and circuit lower bounds.

\paragraph*{\blue{Lower bounds for small-depth circuits}}

A fruitful line of work from the '80s~\cite{FSS81,Ajt83,Yao85,Has86} managed to prove strong lower bounds against \blue{small-depth} circuits, but using a surprisingly simple function: parity. In particular, the result that culminated from these works was that any depth-$d$ circuit computing $\Parity_n$ must have size $2^{\Omega\big(n^{\frac{1}{d-1}}\big)}$, which is superpolynomial for any $d = o(\log n/\log\log n)$. These bounds are optimal for parity---an extension of the divide-and-conquer scheme mentioned earlier matches this bound. As we now elaborate, these remain essentially  the strongest (explicit) lower bounds against small-depth circuits we have for any function, even in the case of $d=3$, despite the fact that counting arguments give lower bounds of $\Omega(2^n/n)$, and we expect that circuits for hard functions like $\SAT$ must also have size $2^{\Omega(n)}$. 

\paragraph*{Depth-3 circuits and majority.} The problem of improving the state of the art for depth-3 circuits $(2^{\Omega(\sqrt{n})}$ for $\Parity_n$) in particular has received significant attention as one of the simplest restricted models that are poorly understood. Stronger depth-3 bounds are likely to imply stronger small-depth bounds in general, and \emph{much} stronger bounds of the form $2^{\omega(n/\log\log n)}$ would also give functions that cannot be computed by linear-size log-depth circuits due to a classical result of Valiant \cite{Val83}. %
For a detailed exposition of the  motivations for depth-3 and the attempts to understand the model, see \cite[Chapter~11]{Juk12}.

Of the functions that could prove stronger depth-3 lower bounds, majority, being such a basic and simple-to-understand function, is a particularly enticing candidate.  The divide-and conquer construction analogous to that for parity gives depth-3 circuits of size $2^{O(\sqrt{n\log n})}$, and the same $2^{\Omega(\sqrt{n})}$ lower bound for parity also applies to majority.\footnote{For general depth $d$, the upper bound is $2^{O\big(n^{\frac{1}{d-1}}\cdot (\log n)^{1-\frac{1}{d-1}}\big)}$ \cite{KPPY84} and the lower bound is $2^{\Omega\big(n^{\frac{1}{d-1}}\big)}$ once again.} In light of this gap, H{\aa}stad, Jukna, and Pudl{\'a}k~\cite{HJP93} posed the following challenge.
 
\begin{openproblem}[\cite{HJP93}]\label{op:depth3lbs}
Find an explicit function that requires depth-3 circuits of size $2^{\omega(\sqrt{n})}$.
In particular, does $\Majority_n$ require depth-3 circuits of that size?
\end{openproblem}

Despite this natural candidate function, all of the improvements on the lower bounds for \mbox{depth-3} circuits have still been of the form $2^{c\sqrt{n}}$ for successively larger $c$. \cite{HJP93} found an innovative method of proving lower bounds for circuits from the ``top down'', and were able to get constants $c = 0.618$ and $0.849$ for parity and majority respectively.
Paturi, Pudl{\'a}k, and Zane~\cite{PPZ97} improved the constant to the optimal $c=1$ for parity, and later
the same authors along with Saks~\cite{PPSZ05}
obtained the constant
$c = 1.282$ for the membership function of an error-correcting code.
Getting a super-constant improvement over this state of the art has been a major frontier of circuit complexity for decades. 

For majority, there have also been attempts at improved \emph{upper bounds}. \cite{Wol06} proposed a probabilistic construction of a depth-3 circuit computing $\Majority_n$ with size $2^{O(\sqrt{n})}$, but the construction turned out to be mistaken.\footnote{Briefly, in the notation of \cite{Wol06}, it requires that $kn/c < n$, and so $k < c$, but $c$ then takes on all the values of $n, n/2, n/3, \ldots, 1$ and $k \approx \sqrt{n}$. The third author thanks Srikanth Srinivasan \cite{Sri15} who informed him about this gap in the proof.}

\paragraph*{The need for new techniques.}  
One view of why previous techniques have fallen short on resolving \Cref{op:depth3lbs} is that at their core, they use sensitivity as the key complexity measure  for which small depth circuits are weak (see~\cite{Bop97,Ros18,MW19}). With respect to sensitivity, of course $\Parity_n$ is the most complex function because it has sensitivity $n$ at every input.
The fact that parity is the hardest should suggest why we have struggled to get lower bounds stronger than those for parity.
More concretely, these techniques do not establish bounds stronger than $2^{\Omega(s(f)^{\frac{1}{d -1}})}$  where $s(f)$ is the sensitivity of $f:\zo^n\to\zo$. But $s(f)\leq n$, so this leaves us stuck the current state of the art ($2^{\Omega(\sqrt{n})}$ for depth-3).

To push beyond our current small-depth circuit lower bounds, there is a need for new techniques.
In particular, we need to make use of complexity measures beyond sensitivity, where parity is no longer the hardest function.
Moreover, to solve \Cref{op:depth3lbs}, such a complexity measure needs to be one where majority is demonstrably harder than parity.
The notion of composition complexity and the techniques of this paper meet both of these demands (as shown by \Cref{thm:main}).
We are hopeful that these techniques (described in \Cref{sec:techniques}) can be extended to prove stronger lower bounds in more general settings.

\paragraph*{Composition complexity and depth-3 circuits.} If $\CC_k(f) \leq m$, then we can write $f=  h(g_1, \ldots, g_m)$ where each of the functions $g_j$ only needs to query $k$ variables.
But then  we can write $h$ as a DNF (or CNF) of size $2^m$, and we can write the inner functions $g_j$ as CNFs (or DNFs) of size $2^k$.
In this form we have a depth-3 circuit for $f$ of size
\begin{equation}
\label{eq:cc-to-depth3}
2^m + m\cdot 2^k
\leq 
2^{O(m + k)}
\end{equation}
with bottom fan-in $k$.
The best-known depth-3 circuits for computing $\Maj_n$ are obtained in precisely this manner,
by using the bound
$\CC_k(\Maj_n) \leq O(\frac{n}{k} \log k)$ and setting $k \coloneqq \sqrt{n\log n}$.

It follows that in order to prove that $\Majority_n$ requires depth-3 circuits of size $2^{\Omega(\sqrt{n \log n})}$, one must first prove that $\max(\CC_k(f), k) \geq \Omega(\sqrt{n \log n})$.
This is implied directly by \Cref{thm:main}, which we view as a key \orange{first} step towards proving stronger depth-3 lower bounds. In particular, it shows that if one wanted to construct a depth-3 circuit for $\Majority_n$ of size $2^{o(\sqrt{n \log n})}$, it would have to look very different from the current divide-and-conquer strategy.
In \Cref{subsec:depth3-plan} we outline some ways in which our results could be extended beyond the model we consider to depth-3.

\subsection{Our techniques}\label{sec:techniques}

In this section, we briefly describe the techniques we use to prove \Cref{thm:main}, and highlight some aspects that we feel are particularly interesting.

\paragraph*{Information theory.}
While there have been some attempts to incorporate information theory into the toolbox of boolean function lower bounds~\cite{NW95,GMWW14}, these techniques remain uncommon.
Our proof crucially uses
\emph{mutual information} to measure information flow from the input variables to the inner functions.

Our key insight can be summed up as the counterintuitive maxim ``the less it is queried the more it is revealed''.
More concretely, suppose we can compute the Hamming weight function as $h(g_1, \ldots ,g_m)$, then we show that
if few of the inner functions $g_1, \ldots, g_m$ query a particular input variable $x_i$, then the output of the inner functions must reveal a lot of information about $x_i$'s value---that is, one can guess $x_i$ better than random chance based on $g_1(x), \ldots, g_m(x)$.

\begin{lemma}[key insight, informally]
\label{lemma:insight}
Suppose that $x_i$ is queried by at most $q$ of the inner functions $g_1, \ldots, g_m$.
Then $\I\sqb{\bX_i:g_1(\bX), \ldots, g_m(\bX)} \geq 2^{-O(q)}$.
\end{lemma}

One of the novelties of our proof is that the information flow can be distilled so cleanly as the above lemma, and that tight lower bounds follow quite easily from it (see \Cref{sec:hw-lb}).
It is natural to wonder whether this approach can be generalized to stronger models of computation.

\paragraph*{From multi-output functions to binary-output functions.}
Another intriguing aspect of our proof is that in
 proving the lower bound for majority,
it turns out to be easiest to first prove a lower bound for the Hamming weight function $\HW_n: \zo^n\to \set{0, \ldots, n}$, which maps $x \mapsto |x|$ (i.e., the number input bits that are 1).
Notably, this is not a binary-output function, but rather a ``multi-output'' function (its output takes $\ceil{\log (n+1)}$ bits), and the proof of our information-theoretic \Cref{lemma:insight} fundamentally uses this larger output space.

The way in which we extend the lower bound from $\HW_n$ to $\Maj_n$ is also worth noting.
In \Cref{sec:hw-to-maj}, we give a general framework for, in a sense, forcing binary-output functions to become multi-output.
We use ``control variables'' to manipulate the construction $h(g_1, \ldots, g_m)$ into telling us more about the input, and ``buffer variables'' to avoid incurring a blowup in how many inner functions are necessary.
This allows us to show that an efficient construction for $\Maj_n$ would imply a similarly efficient construction computing the Hamming weight on a large fraction of inputs, and we can then conclude with a slightly more general version of \Cref{lemma:insight}.

Perhaps there is more to be found in this direction.
Could the frontier of circuit lower bounds be pushed further by first proving lower bounds for multi-output functions, and then bootstrapping these to get lower bounds against the usual single-output functions?
Indeed, our proof technique suggests that proving more lower bounds for multi-output functions could be a valuable endeavor, even when those lower bounds do not seem to immediately lead to lower bounds for binary-output functions we traditionally study.

\paragraph*{Natural proofs.}
The proof of our key lemma in \Cref{sec:proof-main} is tailored specifically to the Hamming weight function.\footnote{At a high level, it uses the facts that the possible Hamming weight values $0, \ldots, n$ are a completely ordered set
and that flipping a bit from $0$ to $1$ increases the Hamming weight by one,
in order to find one weight $\myw$ for which the corresponding inputs are biased on a given coordinate.}
While this could be seen as a limitation, it can also be seen as a strength.
Indeed, Razborov and Rudich \cite{RR97} showed that lower bound arguments
cannot apply to too broad a range of functions, assuming that pseudorandom functions exist.
Given that our lower bound argument only works for the Hamming weight and closely related functions like the majority function, it does not seem to fall within the natural proofs framework.

\subsection{Related work}\label{sec:related}

\paragraph*{Hrube\v{s}-Rao and Nechiporuk's method.}
Hrubeš and Rao~\cite{HR15} gave a function $f$ for which
$\CC_k(f) \geq n^{\Omega(1 - k/n)}$.\footnote{Their paper denotes $k$-composition complexity as $C^2_k$ instead of $\CC_k$.}
Their proof draws inspiration from Nechiporuk's method \cite{Nec66},
which gives lower bounds for functions $f$ for which one
can split the variables into disjoint sets $S_1, \ldots, S_\ell$ 
such that $f$ has many distinct subfunctions on each $S_r$ ($r \in [\ell]$).

Their lower bound is quantitatively stronger than
ours,
but just like other bounds obtained from
Nechiporuk's method,
it only applies
to a limited set of functions specially created for that purpose,
and says nothing about many basic functions like majority.
Nechiporuk's method is unable to prove lower bounds for majority because its subfunctions are all threshold functions, of which there are only a handful.\footnote{Indeed, there is an exactly analogous situation with branching programs. While Nechiporuk's method can establish strong branching program lower bounds (see \cite[Theorem~2]{Raz91}, attributed to Beame and Cook), \cite{AM86,BPRS90} had to introduce new techniques to prove lower bounds for majority.} Moreover, the functions that \cite{HR15} uses have no bearing on stronger lower bounds for depth-3 circuits, since their functions actually have depth-2 circuits of size $n^{O(\log n)}$.

It is interesting to contrast our techniques with Nechiporuk's method. At a very high level, Nechiporuk's method considers what one can infer about the function after fixing several variables (particularly, how many subfunctions remain), whereas our method considers what one can infer about the variables given the output of several of the inner functions. This difference in perspective is key in our ability to get tight lower bounds in a case where Nechiporuk's method only gives trivial bounds.

\paragraph*{Lower bounds for canonical boolean circuits.} Goldreich and Wigderson \cite{GW20} recently introduced a new restricted of model of ``canonical'' boolean circuits, with the hope of proving $2^{\omega(\sqrt{n})}$ lower bounds for this model. Their model is inspired by the structure of optimal depth-3 circuits for $\Parity_n$, which dissects the computation into disjoint parities of smaller arity (over $\sqrt{n}$ variables). \cite{GW20} proposes a generalization of this construction, where one aims to represent a multi-linear function as a depth-2 circuit where the gates are \emph{arbitrary multi-linear functions of small arity}. They propose proving strong lower bounds in this model as a ``sanity check'' for proving better general depth-3 circuit bounds (\Cref{op:depth3lbs})---to do the latter, one must necessarily do the former as well.

Our approach can be viewed in a very similar light.
In fact, our model is strictly stronger:
we consider gates that compute arbitrary \emph{boolean functions} of small arity, not just functions that are multi-linear over $\mathrm{GF}(2)$.
In the same way, our model serves as a sanity check for \Cref{op:depth3lbs}---any proof that depth-3 circuits for $\Maj_n$ require size $2^{\Omega(\sqrt{n\log n})}$ must prove \Cref{thm:main} as well.

Strong lower bounds have indeed been proven in the model introduced by \cite{GW20}.  Goldreich and Tal~\cite{GT18} proved a lower bound of $2^{\tilde{\Omega}(n^{2/3})}$.~\cite{GT20} also proved lower bounds of $2^{\tilde{\Omega}(n^{3/8})}$ for canonical depth-4 circuits, an improvement over the $2^{\Omega(n^{1/3})}$ bound that is known for the general depth-4 circuits computing $\Parity_n$.  In our model where the gates can compute arbitrary boolean functions of small arity, such strong bounds are not possible for the $\Maj_n$ function, and we pin down exactly the right bounds in this case.

\paragraph*{Majority as a majority of majorities.} There has been interest \cite{KP19,EGMR20,HNRRY19,KM18,Pos17} in the optimal ways of computing majority as a composition of functions $h(g_1, \ldots, g_m)$ in the restricted model, where the $h$ and $g_1, \ldots, g_m$ are majority functions of smaller fan-in ($\Maj_{\leq k}$ for some $k$). \cite{KP19} showed that $k\geq n^{0.7}$ was necessary, \cite{EGMR20} improved this to $k \geq n^{0.8}$, and \cite{HNRRY19} further improved this to $k \geq n/2 - o(n)$. In terms of upper bounds, \cite{KM18} gave $k \leq n+2$ for odd $n \geq 7$, and \cite{Pos17} improved this to $k \leq \frac{2}{3}n+4$ (for all $n$).

Similar to the state of canonical boolean circuits discussed above, this setting is more restrictive, so it makes sense that their lower bounds are stronger than \Cref{thm:main}. 
It would be interesting to see if our techniques could apply to these more restricted models as well.

%% file: fig-allowed-gj.tex
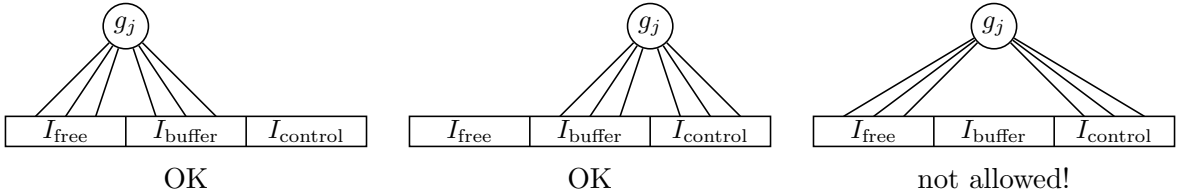
\begin{figure}[h]
\centering
\begin{tabu}{c}
\begin{tikzpicture}[scale=0.4]
\draw (0,0) rectangle (12,1);
\draw (4,0) -- (4,1);
\draw (8,0) -- (8,1);
\node at (2,0.5) {$\Ifree$};
\node at (6,0.5) {$\Ibuffer$};
\node at (10,0.5) {$\Icontrol$};

\node[circle,draw,inner sep=2pt] (g) at (4,4) {$g_j$};

\draw (g) -- (1,1);
\draw (g) -- (2,1);
\draw (g) -- (3,1);
\draw (g) -- (5,1);
\draw (g) -- (6,1);
\draw (g) -- (7,1);
\end{tikzpicture}\\
OK\\
\end{tabu}
\begin{tabu}{c}
\begin{tikzpicture}[scale=0.4]
\draw (0,0) rectangle (12,1);
\draw (4,0) -- (4,1);
\draw (8,0) -- (8,1);
\node at (2,0.5) {$\Ifree$};
\node at (6,0.5) {$\Ibuffer$};
\node at (10,0.5) {$\Icontrol$};

\node[circle,draw,inner sep=2pt] (g) at (8,4) {$g_j$};

\draw (g) -- (5,1);
\draw (g) -- (6,1);
\draw (g) -- (7,1);
\draw (g) -- (9,1);
\draw (g) -- (10,1);
\draw (g) -- (11,1);
\end{tikzpicture}\\
OK\\
\end{tabu}
\begin{tabu}{c}
\begin{tikzpicture}[scale=0.4]
\draw (0,0) rectangle (12,1);
\draw (4,0) -- (4,1);
\draw (8,0) -- (8,1);
\node at (2,0.5) {$\Ifree$};
\node at (6,0.5) {$\Ibuffer$};
\node at (10,0.5) {$\Icontrol$};

\node[circle,draw,inner sep=2pt] (g) at (6,4) {$g_j$};

\draw (g) -- (1,1);
\draw (g) -- (2,1);
\draw (g) -- (3,1);
\draw (g) -- (9,1);
\draw (g) -- (10,1);
\draw (g) -- (11,1);
\end{tikzpicture}\\
not allowed!\\
\end{tabu}

\caption{The inner functions $g_j$ are allowed to query variables in both $\Ifree$ and $\Ibuffer$, or variables in both $\Ibuffer$ and $\Icontrol$, but not variables in both $\Ifree$ and $\Icontrol$.}
\label{fig:allowed-gj}
\end{figure}

%% file: fig-free-buffer-control.tex
\begin{figure}[h]
\centering 

\begin{tikzpicture}[scale=0.4, font=\small]
\draw (0,0.5) rectangle (14,8.5);
\node[above] at (7,8.5) {$[n]$};

\draw (10,4) circle (3);
\node[above] at (10,7) {$I'$};

\draw (10,5) circle (1.5);
\node at (10,5) {$\Icontrol$};
\node at (10,2.5) {$\Ibuffer$};
\node at (3.5,4.25) {$\Ifree$};
\end{tikzpicture}

\caption{A schematic view of the sets of variables $\Ifree$, $\Ibuffer$, and $\Icontrol$.}
\label{fig:free-buffer-control}
\end{figure}